\documentclass{llncs}

\usepackage{amsmath,amsfonts,amssymb,enumerate}

\usepackage{graphicx,color}
\usepackage{hyperref}
\usepackage{boxedminipage}
\newtheorem{observation}{Observation}

\newcommand{\vd}{{\sf vd}}
\newcommand{\ed}{{\sf ed}}
\newcommand{\ea}{{\sf ea}}

\newcommand{\tw}{{\mathbf{tw}}}

\newcommand{\poly}{\textrm{\rm poly}}

\DeclareMathOperator{\operatorClassNP}{{\sf NP}}
\newcommand{\classNP}{\ensuremath{\operatorClassNP}}
\DeclareMathOperator{\operatorClassCoNP}{{\sf coNP}}
\newcommand{\classCoNP}{\ensuremath{\operatorClassCoNP}}
\DeclareMathOperator{\operatorClassFPT}{{\sf FPT}}
\newcommand{\classFPT}{\ensuremath{\operatorClassFPT}}
\DeclareMathOperator{\operatorClassW}{{\sf W}}
\newcommand{\classW}[1]{\ensuremath{\operatorClassW[#1]}}

\newcounter{ctrclaim}[theorem]
\renewcommand{\thectrclaim}{\Alph{ctrclaim}}
\newcommand\displaycase[1]{{\bf #1}}
\newcommand{\clm}[1]{\medskip\phantomsection\refstepcounter{ctrclaim}\noindent\displaycase{Claim \thectrclaim.}{\em #1}\\}

\begin{document}

\title{Editing to a Planar Graph of Given Degrees\thanks{
An extended abstract of this paper appeared in the proceedings of CSR 2015~\cite{DGHPT15}.
The first and fourth author were supported by EPSRC Grant EP/K025090/1.
The research of the second author has received funding from the European Research Council under the European Union's Seventh Framework Programme (FP/2007-2013)/ERC Grant Agreement n. 267959.
The research of the fifth author was co-financed by the European Union (European Social Fund ESF) and Greek national funds through the Operational Program ``Education and Lifelong Learning'' of the National Strategic Reference Framework (NSRF) - Research Funding Program: ARISTEIA II.}}

\author{
Konrad K. Dabrowski\inst{1}, 
Petr A. Golovach\inst{2}, 
Pim van 't Hof\inst{3},
\\
Dani{\"e}l Paulusma\inst{1}, and
Dimitrios M. Thilikos\inst{4}
}

\institute{
School of Engineering and Computing Sciences, Durham University, United Kingdom. E-mail: \texttt{\{konrad.dabrowski,daniel.paulusma\}@durham.ac.uk}
\and
Department of Informatics, University of Bergen, Norway. E-mail: {\tt{petr.golovach@ii.uib.no}}
\and
School of Built Environment, Rotterdam University of Applied Sciences, Rotterdam, the Netherlands. E-mail: \texttt{p.van.t.hof@hr.nl}
\and
Computer Technology Institute and Press ``Diophantus'', Patras, Greece,
Department of Mathematics, National and Kapodistrian University of Athens, Athens, Greece and AlGCo project-team, CNRS, LIRMM, Montpellier, France.
E-mail: {\tt{sedthilk@thilikos.info}}
}

\maketitle

\begin{abstract}
We consider the following graph modification problem.
Let the input consist of a graph $G=(V,E)$, a weight function $w\colon V\cup E\rightarrow \mathbb{N}$, a cost function $c\colon V\cup E\rightarrow \mathbb{N}$ and a degree function $\delta\colon V\rightarrow \mathbb{N}_0$, together with three integers $k_v, k_e$ and~$C$. The question is whether we
can delete a set of vertices of total weight at most~$k_v$ and a set of edges of total weight at most~$k_e$ so that the total cost of the deleted elements is at most~$C$ and every non-deleted vertex~$v$ has degree~$\delta(v)$ in the resulting graph~$G'$.
We also consider the variant in which~$G'$ must be connected.
Both problems are known to be \classNP-complete and \classW{1}-hard when parameterized by $k_v+k_e$.
We prove that, when restricted to planar graphs, they stay \classNP-complete but have polynomial kernels when parameterized by $k_v+k_e$.
\end{abstract}

\section{Introduction}
Graph modification problems capture a variety of graph-theoretic problems and are well studied in algorithmic graph theory. The aim is to modify some given graph~$G$ into some other graph~$H$ that satisfies a {\it certain property}
by applying a bounded number of operations from a set~$S$ of {\it prespecified graph operations}.
Well-known graph operations are the edge addition, edge deletion and vertex
deletion, denoted by $\ea, \ed$ and~$\vd$, respectively.
For example, if $S=\{\vd\}$ and~$H$ must be a clique or independent set then we obtain the basic problems
{\sc Clique} and {\sc Independent Set}, respectively. To give a few more examples, if~$H$ must be a forest and
$S=\{\ed\}$ or $S=\{\vd\}$ then we obtain the problems {\sc Feedback Edge Set} and {\sc Feedback Vertex Set}, respectively. As discussed in detail later, it is also common to consider sets~$S$ consisting of more than one graph operation.

A property is {\it hereditary} if it holds for any induced subgraph of a graph that satisfies it, and a property is {\it non-trivial} if it is both true for infinitely many graphs and false for infinitely many graphs.
A classic result of Lewis and Yannakakis~\cite{LewisY80} is that a vertex deletion problem is
\classNP-hard for any property that is hereditary and non-trivial.
In an earlier paper Yannakakis~\cite{Yannakakis78} also showed that the edge deletion problem is \classNP-complete for several properties, such as being planar or outer-planar.
Natanzon, Shamir and Sharan~\cite{NatanzonSS01} and Burzyn, Bonomo and Dur{\'a}n~\cite{BurzynBD06} proved that the graph modification problem is \classNP-complete when $S=\{\ea,\ed\}$ and the desired property is to belong to some hereditary graph class for a variety of such graph classes.

When a problem turns out to be \classNP-hard, a possible next step might be to consider it in the more refined framework offered by {\it parameterized complexity}. This is certainly an appropriate direction to follow for graph modification problems, because the bound on the total number of permitted operations is a natural parameter~$k$.
Cai~\cite{Cai96} proved that for this parameter the graph modification problem is \classFPT\ if $S=\{\ea,\ed,\vd\}$ and
the desired property is to belong to any fixed graph class characterized by a finite set of forbidden induced subgraphs.
Khot and Raman~\cite{KhotR02} determined all non-trivial hereditary properties for which the vertex deletion problem is \classFPT\ on $n$-vertex graphs with parameter $n-\nobreak k$ and proved that for all other such properties the problem is \classW{1}-hard (when parameterized by $n-\nobreak k$).

From the aforementioned results we conclude that the graph modification problem has been thoroughly studied for hereditary properties. However, for other types of properties, much less is known.
Dabrowski et al.~\cite{DGHP14} combined previous results~\cite{BoeschST77,CaiY11,CyganMPPS14} with new results to
classify the (parameterized) complexity of the problem of modifying the input graph into a connected graph where
each vertex has some prescribed degree parity for all $S\subseteq \{\ea,\ed,\vd\}$.

In this paper we consider the case when the vertices of the resulting graph must satisfy some prespecified degree constraints (note that such properties are non-hereditary, so the results of Lewis and Yannakakis do not apply to this case). Before presenting our results, we briefly discuss the known results and the general framework they fall under.

Moser and Thilikos in~\cite{MoserT09} and Mathieson and Szeider~\cite{MathiesonS12}
initiated an investigation into the parameterized complexity of such graph modification problems.
In particular, Mathieson and Szeider~\cite{MathiesonS12} introduced the following general problem.
\begin{center}
\begin{boxedminipage}{.99\textwidth}
\textsc{Degree Constraint Editing($S$)}\\
\begin{tabular}{ r p{0.8\textwidth}}
\textit{~~~~Instance:} & A graph~$G$, integers $d,k$ and a function
                                  $\delta\colon V(G)\rightarrow\{1,\ldots,d\}$.\\
\textit{Question:} & Can~$G$ be modified into a graph~$G'$ such
that $d_{G'}(v)=\delta(v)$ for each $v\in V(G')$ using at most~$k$ operations
from the set~$S$?
\end{tabular}
\end{boxedminipage}
\end{center}
Mathieson and Szeider~\cite{MathiesonS12} classified the parameterized complexity of this problem for
$S\subseteq \{\ea,\ed,\vd\}$.
In particular they showed the following results.
If $S\subseteq \{\ea,\ed\}$ then the problem is polynomial-time solvable.
If $\vd\in S$ then the problem is \classNP-complete, \classW{1}-hard with parameter~$k$ and \classFPT\ with parameter~$d+k$.
Moreover, they proved that the latter result holds even for a more general version, in which the vertices and edges have costs and
the desired degree for each vertex should be in some given subset of $\{1,\ldots,d\}$.
If $S\subseteq \{\ed,\vd\}$, they proved that the problem has a polynomial kernel when parameterized by~$d+k$.
Golovach~\cite{Golovach14a} considered the cases $S=\{\ea,\vd\}$ and $S=\{\ea,\ed,\vd\}$ and proved (amongst other results) that for these cases
the problem has no polynomial kernel 
when parameterized by $d+k$
unless $\classNP\subseteq\classCoNP/\text{\rm poly}$.
Froese, Nichterlein and Niedermeier~\cite{FroeseNN14} gave more kernelization results for {\sc Degree Constraint Editing($S$)}.
Golovach~\cite{Golovach14} introduced a variant of {\sc Degree Constraint Editing($S$)} in which we additionally insist that the resulting graph must be {\it connected}. He proved that, for $S=\{\ea\}$, this variant is \classNP-complete, \classFPT\ 
when parameterized by~$k$, and has a polynomial kernel when parameterized by $d+k$.
The connected variant is readily seen to be \classW{1}-hard when $\vd\in S$
by a straightforward modification of the proof of the \classW{1}-hardness result for {\sc Degree Constraint Editing$(S)$}, when $\vd\in S$, as given by Mathieson and Szeider~\cite{MathiesonS12}.

In the light of the above \classNP-completeness and \classW{1}-hardness results (when $\vd \in S)$ it is natural to
restrict the input graph~$G$ to a special graph class.
Hence, inspired by the above results, we consider the set $S=\{\ed,\vd\}$ and study weighted versions of both variants (where we insist that the resulting graph is connected and where we don't) of these problems for {\it planar} input graphs.
In fact the problems we study are even more general.
The problem variant not demanding connectivity is defined as follows.

\begin{center}
\begin{boxedminipage}{.99\textwidth}
\textsc{Deletion to a Planar Graph of Given Degrees (DPGGD)}\\
\begin{tabular}{ r p{0.8\textwidth}}
\textit{~~~~Instance:} & A planar graph $G=(V,E)$, integers $k_v,k_e,C$ and functions
                                 $\delta\colon V\rightarrow \mathbb{N}_0$, $w\colon V\cup E\rightarrow \mathbb{N}$,
                                   $c\colon V\cup E\rightarrow \mathbb{N}_0$.\\
\textit{Question:} & Can~$G$ be modified into a graph~$G'$ by deleting a set $U\subseteq V$ with $w(U)\leq k_v$ and a set $D\subseteq E$ with $w(D)\leq k_e$ such that $c(U\cup D)\leq C$ and $d_{G'}(v)=\delta(v)$ for $v\in V(G')$?
\end{tabular}
\end{boxedminipage}
\end{center}
In the above problem,~$w$ is the \emph{weight} and~$c$ is the \emph{cost} function. The question is whether it is possible to delete vertices and edges of total weight at most~$k_v$ and~$k_e$, respectively, so that the total cost of the deleted elements is at most~$C$ and the obtained graph satisfies the degree restrictions prescribed by the given function~$\delta$.

The second problem we consider is the variant of DPGGD, in which the desired graph~$G'$ must be connected. We call this variant the \textsc{Deletion to a Connected Planar Graph of Given Degrees} problem (DCPGGD).

\medskip
\noindent
{\bf Our Results}.
We note that \textsc{DPGGD} is \classNP-complete even if $\delta\equiv 3, w\equiv 1,\allowbreak c\equiv\nobreak 0$ and $k_v=|V(G)|-1$, and \textsc{DCPGGD} is \classNP-complete
even if $\delta\equiv 2,\allowbreak w\equiv\nobreak 1,\allowbreak c\equiv 0$
and $k_v=0$.
These observations follow directly from the respective facts that both testing
whether a planar graph of degree at most~$7$ has a non-trivial cubic
subgraph~\cite{Stewart94} is \classNP-complete and testing whether a cubic planar graph has a
Hamiltonian cycle~\cite{GareyJT76} is \classNP-complete.
In contrast to the aforementioned \classW{1}-hardness results for general graphs,
our two main results are that both DPGGD and DCPGGD have polynomial kernels when parameterized by~$k_v+\nobreak k_e$. Note that the integer~$C$ is neither a constant nor a parameter but part of the input.
In order to obtain our results we first show that both problems are
polynomial-time solvable
for any graph class of bounded treewidth. We then use the \emph{protrusion decomposition/replacement} techniques introduced by Bodlaender at al.~\cite{BodlaenderFLPST09} (see~\cite{BodlaenderFLPST09a} for the full text). These techniques were successfully used for various problems on sparse
graphs~\cite{FominLST12,GarneroPST14,GarneroST14,KimLPRRSS13}. We stress that \textsc{DPGGD} and \textsc{DCPGGD} do not fit in the meta-kernelization framework of Bodlaender at al.~\cite{BodlaenderFLPST09}. Hence our approach is, unavoidably, problem-specific.

\section{Preliminaries}\label{sec:defs}
\noindent
All graphs in this paper are finite, undirected and without loops or multiple
edges. The vertex set of a graph~$G$ is denoted by~$V(G)$ and the edge set is denoted by~$E(G)$.
For a set $X\subseteq V(G)$, we let~$G[X]$ denote the subgraph of~$G$ induced by~$X$. We write $G-X=G[V(G)\setminus X]$; we allow the case where $X\not\subseteq V(G)$.
If $X=\{x\}$, we may write $G-x$ instead.
For a set $L\subseteq E(G)$, we let
$G-L$ be the graph obtained from~$G$ by deleting all edges of~$L$. If $L=\{e\}$ then we may write $G-e$ instead.
For $v\in V(G)$, let $E_G(v)=\{e\in E(G)\mid e\text{ is incident to }v\}$.
For $X\subseteq V(G)$, let $E_G(X)=\bigcup_{v\in X}E_G(v)$.
For $e\in E(G)$ with $e=uv$, let $V(e)=\{u,v\}$.
For a set $L\subseteq E(G)$ let $V(L)=\cup_{e\in L}V(e)$.

Let~$G$ be a graph.
For a vertex~$v$, we let~$N_G(v)$ denote its
\emph{(open) neighbourhood}, that is, the set of vertices adjacent to~$v$.
The \emph{degree} of a vertex~$v$ is denoted by $d_G(v)=|N_G(v)|$.
For a set $X\subseteq V(G)$, we write $N_G(X)=(\bigcup_{v\in X}N_G(v))\setminus X$.
The \emph{closed neighbourhood} $N_G[v]=N_G(v)\cup \{v\}$, and for a non-negative integer~$r$, $N_G^r[v]$ is the set of vertices at distance at most~$r$ from~$v$;
note that $N_G^0[v]=\{v\}$ and that $N_G^1[v]=N_G[v]$.
For a set $X\subseteq V(G)$ and a positive integer~$r$, let $N_G^r[X]=\bigcup_{v\in X}N_G^r[v]$.
For a positive integer~$r$, a set $X\subseteq V(G)$ is an \emph{$r$-dominating} set of~$G$ if $V(G)\subseteq N_G^r[X]$.
For a set $X\subseteq V(G)$, $\partial_G(X)=X\cap N_G(V(G)\setminus X)$ is the \emph{boundary} of~$X$ in~$G$.

A \emph{tree decomposition} of a graph~$G$ is a pair $(\mathcal{X},T)$ where~$T$
is a tree and $\mathcal{X}=\{X_i \mid i\in V(T)\}$ is a collection of subsets (called {\em bags})
of~$V(G)$ such that
\begin{enumerate}[(i)]
\item $\bigcup_{i \in V(T)} X_i = V(G)$,
\item for each edge $xy \in E(G)$, $x,y\in X_i$ for some $i\in V(T)$, and
\item for each $x\in V(G)$, the set $\{ i \mid x \in X_i \}$ induces a connected subtree of~$T$.
\end{enumerate}
The \emph{width} of a tree decomposition $(\{ X_i \mid i \in V(T) \},T)$ is $\max_{i \in V(T)}\,\{|X_i| - 1\}$. The \emph{treewidth} of a graph~$G$ (denoted~$\tw(G)$) is the minimum width over all tree decompositions of~$G$.
A tree decomposition $(\mathcal{X},T)$ of a graph~$G$ is \emph{nice}, if~$T$ is a rooted binary tree such that the nodes of~$T$ are of four types:
\begin{enumerate}[(i)]
\item a \emph{leaf node}~$i$ is a leaf of~$T$ with $X_i=\emptyset$;
\item an \emph{introduce node}~$i$ has one child~$i'$ with $X_i=X_{i'}\cup\{v\}$ for some vertex $v\in V(G)$;
\item a \emph{forget node}~$i$ has one child~$i'$ with $X_i=X_{i'}\setminus\{v\}$ for some vertex $v\in V_G$; and
\item a \emph{join node}~$i$ has two children~$i'$ and~$i''$ with $X_i=X_{i'}=X_{i''}$,
\end{enumerate}
and, moreover, the root~$r$ is a forget node with $X_r=\emptyset$.
Kloks~\cite{Kloks94} proved that every tree decomposition of a graph can be converted in linear time to a nice tree decomposition of the same width such that the size of the obtained tree is linear in the size of the original tree.

We need the following known observation, which is valid for every planar bipartite graph~$G$ in which the vertices of one partition class~$V_2$ have degree at least~$3$ (in order to prove this, note that $3|V_2|\leq\sum_{v\in V_2}d_G(v)=|E(G)|\leq 2|V(G)|-4$, as~$G$ is bipartite and planar).

\begin{lemma}\label{lem:bound-bip}
Let~$V_1$ and~$V_2$ be bipartition classes of a planar bipartite graph~$G$ such that $d_G(v)\geq 3$ for every $v\in V_2$ and~$V_2$ is non-empty.
Then $|V_2|\leq 2|V_1|-4$.
\end{lemma}

\noindent
{\bf Protrusion decompositions.} For a graph~$G$ and a positive integer~$r$, a set $X\subseteq\nobreak V(G)$ is an \emph{$r$-protrusion} of~$G$ if
$|\partial_G(X)|\leq r$ and $\tw(G[X])\leq r$.
For positive integers~$s$ and~$s'$, an $(s,s')$-protrusion decomposition of a graph~$G$ is a partition $\Pi=\{R_0,\ldots,R_p\}$ of~$V(G)$ such that
\begin{enumerate}[(i)]
\item $\max\{p,|R_0|\}\leq s$,
\item for each $i\in\{1,\ldots,p\}$, $R_i^+=N_G[R_i]$ is an $s'$-protrusion of~$G$, and
\item for each $i\in\{1,\ldots,p\}$, $N_G(R_i)\subseteq R_0 \cap \partial_G[R_i^+]$.
\end{enumerate}
Originally, condition (iii) only demanded that $N_G(R_i)\subseteq R_0$ holds for each $i\in \{1,\ldots,p\}$. However, we can move every vertex
in $N_G(R_i)\setminus \partial_G[R_i^+]$ to~$R_i$ without affecting any of the other properties. Hence we assume without loss of generality that such vertices do not exist and may indeed state condition (iii) as above (which is convenient for our purposes).
The sets $R_1^+,\ldots,R_p^+$ are called the \emph{protrusions} of~$\Pi$.

The following statement is implicit in~\cite{BodlaenderFLPST09a} (see Lemmas~6.1 and 6.2).
\begin{lemma}[\cite{BodlaenderFLPST09a}]\label{lem:protr}
Let~$r$ and~$k$ be positive integers and let~$G$ be a planar graph that has an $r$-dominating set of size at most~$k$.
Then~$G$ has an $(O(kr),O(r))$-protrusion decomposition, which can be constructed in
polynomial time.
\end{lemma}

\medskip
\noindent
{\bf Parameterized Complexity.}
Parameterized complexity is a two dimensional framework
for studying the computational complexity of a problem. One dimension is the input size~$n$ and the other is a parameter~$k$. A problem is said to be \emph{fixed parameter tractable} (or \classFPT) if it can be solved in time $f(k)\cdot n^{O(1)}$ for some function~$f$.
A \emph{kernelization} for a parameterized problem is a polynomial algorithm that maps each instance $(x,k)$ with input~$x$ and parameter~$k$ to an instance $(x',k')$ such that
\begin{enumerate}[(i)]
\item $(x,k)$ is a
yes-instance if and only if $(x',k')$ is a yes-instance, and
\item the size of~$x'$ and~$k'$ is bounded by~$f(k)$ for a computable function~$f$.
\end{enumerate}
The output $(x',k')$ is called a \emph{kernel}. The function~$f$ is said to be the \emph{size} of the kernel. A kernel is \emph{polynomial} if~$f$ is polynomial.
We refer to the books of Downey and Fellows~\cite{DowneyF13},
Flum and Grohe~\cite{FlumG06}, and Niedermeier~\cite{Niedermeierbook06} for detailed introductions to parameterized complexity.

\section{The Polynomial Kernels}\label{sec:ker}

In this section we construct polynomial kernels for \textsc{DPGGD} and
\textsc{DCPGGD}.
We say that a pair $(U,D)$ with $U\subseteq V(G)$ and $D\subseteq E(G)$ is a \emph{solution} for an instance
$(G,k_v,k_e,C,\delta,w,c)$ of
DPGGD if $w(U)\leq k_v$, $w(D)\leq k_e$ and $c(U\cup D)\leq C$ and $G'=G-U-D$ satisfies
$d_{G'}(v)=\delta(v)$ for all $v\in V(G')$. If $(G,k_v,k_e,C,\delta,w,c)$ is an instance of DCPGGD then
$(U,D)$ is a solution if in addition~$G'$ is connected.
Notice that it can happen that $U=V(G)$ for a solution $(U,D)$.

In order to prove our main results, we first need to introduce some additional terminology and prove some structural results.
We say that a solution $(U,D)$ for an instance of \textsc{DPGGD} or \textsc{DCPGGD} is \emph{efficient} if~$D$ has no edges incident to the vertices of~$U$. We say that a solution $(U,D)$ is of \emph{minimum cost} if $c(\hat{U},\hat{D}) \geq c(U,D)$ for every solution $(\hat{U},\hat{D})$.
We make two observations.

\begin{observation}\label{obs:norm-1}
Any yes-instance of \textsc{DPGGD} or \textsc{DCPGGD} has an efficient solution of minimum cost.
\end{observation}

\begin{observation}\label{obs:norm-2}
Let $(G,k_v,k_e,C,\delta,w,c)$ be instance of
\textsc{DPGGD} or \textsc{DCPGGD} that has an efficient solution $(U,D)$.
If $d_G(v)=\delta(v)$ for some $v\in V(G)$ then~$v$ is not incident to an edge of~$D$.
\end{observation}

We say that an instance $(G,k_v,k_e,C,\delta,w,c)$ of
\textsc{DPGGD} (\textsc{DCPGGD} respectively) is \emph{normalized} if
\begin{enumerate}[(i)]
\item for every $v\in V(G)$, $\delta(v)\leq d_G(v)\leq \delta(v)+k_v+k_e$, and
\item every vertex~$v$ in the set $S=\{u\in V(G)\mid d_G(u)=\delta(u)\}$ is adjacent to a vertex in
$\overline{S}=V(G)\setminus S$.
\end{enumerate}

\begin{lemma}\label{lem:normalization}
There is a polynomial-time algorithm that for each instance of \textsc{DPGGD} or \textsc{DCPGGD} either solves the problem or returns an equivalent normalized instance.
\end{lemma}

\begin{proof}
Let $(G,k_v,k_e,C,\delta,w,c)$ be an instance of \textsc{DPGGD}.
To simplify notation, we keep the same notation for the functions $\delta,w,c$ if we delete vertices or edges and do not modify the values of the functions for the remaining elements if this does not create confusion.

We say that a reduction rule is \emph{safe} if by applying the rule we either solve the problem or obtain an equivalent instance.
It is straightforward to see that the following reduction rules are safe.

\noindent
\begin{quote}
{\bf Yes-instance rule.}
If $S=V(G)$ then $(\emptyset,\emptyset)$ is a solution, return a yes-answer and stop.
\end{quote}

\noindent
\begin{quote}
{\bf Vertex deletion rule.}
If~$G$ has a vertex~$v$ with $d_G(v)<\delta(v)$ or
$d_G(v)> \delta(v)+k_v+k_e$, then
delete~$v$ and set $k_v=k_v-w(v)$, $C=C-c(v)$.
If $k_v<0$ or $C<0$,
then stop and return a no-answer.
\end{quote}

Observe that by the exhaustive application of the {\bf vertex deletion rule} and applying the {\bf yes-instance rule} whenever possible, we either solve the problem or we obtain an instance which satisfies (i) of the definition of normalized instances, but where $S\neq V(G)$. Notice that, in particular, the {\bf yes-instance rule} is applied if the set of vertices becomes empty.
To ensure (ii), we apply the following two rules.

\noindent
\begin{quote}
{\bf Contraction rule.} If~$G$ has two adjacent vertices $u,v\in S=\{x\in V(G)\mid d_G(x)=\delta(x)\}$ such that $N_G(v)\subseteq S$,
then we construct the instance $(G',k_v,k_e,C,\delta',w',c')$ as follows.
\begin{itemize}
\item Contract~$uv$. Denote the obtained graph $G'=G/uv$ and let~$z$ be the vertex obtained from~$u$ and~$v$.
\item Set $\delta'(z)=d_{G'}(z)$ and set $\delta'(x)=d_{G'}(x)$ for any $x\in S\setminus\{u,v\}$.
For each $x\in \overline{S}$, set $\delta'(x)=\delta(x)$.
\item Set $w'(z)=w(u)+w(v)$ and $c'(z)=c(u)+c(v)$. For $x\in V(G)\setminus\{u,v\}$, set
$w'(x)=w(x)$ and $c'(x)=c(x)$.
\item For each $xz\in E(G')$, set $w'(xz)=k_e+1$ and $c'(xz)=0$. For all other edges $xy\in E(G')$, set $w'(xy)=w(xy)$ and $c'(xy)=c(xy)$.
\end{itemize}
\end{quote}

Let $(U,D)$ be an efficient solution for $(G,k_v,k_e,C,\delta,w,c)$. By Observation~\ref{obs:norm-2},~$D$ has no edges incident to~$u$ or~$v$. Also either $u,v\in U$ or $u,v\notin U$, because~$u$ and~$v$ are adjacent and $d_G(u)=\delta(u)$ and $d_G(v)=\delta(v)$. Let $U'=(U\setminus\{u,v\})\cup \{z\}$ if $u,v\in U$ and $U'=U$ otherwise. We have that $(U',D)$ is a solution for $(G',k_v,k_e,C,\delta',w',c')$. If $(U',D')$ is an efficient solution for $(G',k_v,k_e,C,\delta',w',c')$, then~$D'$ has no edges incident to~$z$ by Observation~\ref{obs:norm-2}. If $z\in U'$, let $U=(U'\setminus\{z\})\cup \{u,v\}$ and $U=U'$ otherwise. We obtain that $(U,D)$ is a solution for the original instance.

We exhaustively apply the above rule. Assume that it cannot be applied for $(G,k_v,k_e,C,\delta,w,c)$.
Then we have that this instance satisfies~(i) and the following holds:
for any
$v\in S\neq V(G)$,
either~$v$ is adjacent to a vertex in~$\overline{S}$ or~$v$ is an isolated vertex. It remains to deal with isolated vertices.

\noindent
\begin{quote}
{\bf Isolates removal rule.} If~$G$ has an isolated vertex~$v$, then delete~$v$.
\end{quote}

To see that above rule is safe, notice that, because the considered instance satisfies (i), it follows that $\delta(v) \leq d_G(v) = 0$, so $v\in S$.
Clearly, by the exhaustive application of the {\bf isolates removal rule}, we either solve the problem or obtain an instance that satisfies (i) and (ii).\smallskip

Now consider an instance $(G,k_v,k_e,C,\delta,w,c)$ of \textsc{DCPGGD}.

We replace the {\bf yes-instance rule} by the following variant.

\noindent
\begin{quote}
{\bf Yes-instance rule (connected).}
If $S=V(G)$ and~$G$ is connected, then $(\emptyset,\emptyset)$ is a solution, return a yes-answer and stop.
\end{quote}

It is straightforward to verify that the {\bf vertex deletion rule} and the {\bf contraction rule} are safe for this problem. By applying these rules and by the application of the connected variant of the
{\bf yes-instance rule} whenever possible,
we either solve the problem or obtain an equivalent instance that satisfies~(i) and has the property that
for any $v\in S$, either~$v$ is adjacent to a vertex in~$\overline{S}$ or~$v$ is an isolated vertex.
Suppose that $(G,k_v,k_e,C,\delta,w,c)$ satisfies these properties.
Observe that if~$H$ is a component of~$G$, then for any solution $(U,D)$, either $V(H)\subseteq U$ or $V(G)\setminus V(H)\subseteq U$. Therefore, it is safe to apply the following variant of the {\bf isolates removal rule}.

\noindent
\begin{quote}
{\bf Isolates removal rule (connected).} If~$G$ has an isolated vertex~$v$, then if $w(V(G)\setminus \{v\})\leq k_v$ and $c(V(G)\setminus \{v\})\leq C$, then $(V(G)\setminus\{v\},\emptyset)$ is a solution, return a yes-answer and stop. Otherwise, if $w(V(G)\setminus \{v\})>k_v$ or $c(V(G)\setminus \{v\})> C$, delete~$v$ and set $k_v=k_v-w(v)$ and $C=C-c(v)$; if $k_v<0$ or $C<0$, then stop and return a no-answer.
\end{quote}

It is easy to see that if the input graph was planar then the graph formed after
applying the rules above will also be planar.\qed
\end{proof}

\begin{lemma}\label{lem:ds}
If $(G,k_v,k_e,C,\delta,w,c)$ is a normalized yes-instance of \textsc{DPGGD} (\textsc{DCPGGD} respectively) then~$G$ has a $2$-dominating set of size at most $k_v+2k_e$.
\end{lemma}

\begin{proof}
We prove the lemma for \textsc{DPGGD}; the proof for \textsc{DCPGGD} is the same. Let $(G,k_v,k_e,C,\delta,w,c)$ be a normalized yes-instance of the problem. Let $(U,D)$ be a solution and $W=U\cup V(D)$. Clearly, $|W|\leq k_v+2k_e$, because the weights are positive integers. We show that~$W$ is a $2$-dominating set of~$G$.

Let $S=\{v\in V(G)\mid d_G(v)=\delta(v)\}$ and $\overline{S}=V(G)\setminus S$. For any vertex $v\in \overline{S}$, either $v\in U$ or~$v$ is adjacent to a vertex of~$U$ or~$v$ is incident to an edge of~$D$. Hence, $\overline{S}\subseteq N_G[W]$. Let $v\in S$. Because the considered instance is normalized,~$v$ is adjacent to a vertex $u\in \overline{S}$. It implies, that $S\subseteq N_G^2[W]$.
\qed
\end{proof}

The following is a direct consequence of Lemmas~\ref{lem:protr} and~\ref{lem:ds}.

\begin{lemma} \label{lem:ds_pd}
There is a fixed constant~$\alpha$ such that,
if $(G,k_v,k_e,C,\delta,w,c)$ is a normalized yes-instance of \textsc{DPGGD} (\textsc{DCPGGD} respectively), then~$G$ has an $(\alpha (k_v+2k_e),\alpha)$-protrusion decomposition. Moreover, if there is such a decomposition, one can be constructed 
in polynomial time.
\end{lemma}

The next lemma states that, for both \textsc{DPGGD} and
\textsc{DCPGGD}, an optimal solution can be found in polynomial time
on graphs of bounded treewidth.
The proof is based on the standard techniques for dynamic programming over tree decompositions.

\begin{lemma}\label{lem:tw}
\textsc{DPGGD} (\textsc{DCPGGD} respectively) can be solved, and an efficient solution $(U,D)$ of minimum cost can be obtained in $(k_v+k_e)^{O(q)}\cdot \poly(n)$ time (in $(q(k_v+k_e))^{O(q)}\cdot \poly(n)$ time respectively) for
instances
$(G,k_v,k_e,C,\delta,w,c)$ where~$G$ is an $n$-vertex graph of treewidth at most~$q$
and $\delta(v)\leq d_G(v)\leq \delta(v)+k_v+k_e$ for $v\in V(G)$.
\end{lemma}

\begin{proof}
We use more or less standard approach for construction of dynamic programming algorithms for graphs of bounded treewidth.

First, we consider \textsc{DPGGD}.
Let $(G,k_v,k_e,C,\delta,w,c)$ be an instance of the problem where $\tw(G)\leq q$
and $\delta(v)\leq d_G(v)\leq \delta(v)+k_v+k_e$ for all $v\in V(G)$.
We first of all assume that a nice tree decomposition $(\mathcal{X},T)$ of~$G$ with width $t=O(q)$ is given.
To simplify later arguments, we may assume $t \geq 2$.
For this, we may use the algorithm of~\cite{BodlaenderDDFLP13} to obtain an decomposition whose width is at most five times
the optimal in $2^{O(q)}\cdot n$ steps and then convert it to a nice tree decomposition using
the aforementioned results of Kloks~\cite{Kloks94}.

Let~$r$ denote the root of~$T$. For any node $i\in V(T)$, let~$T_i$ denote the subtree of~$T$ induced by~$i$ and its descendants and let $G_i=G[\bigcup_{j\in V(T_i)}X_j]$.
We apply a dynamic programming algorithm over $(\mathcal{X},T)$.

First, we describe the tables that are constructed for the nodes of~$T$. Let $i\in V(T)$. We define ${\bf table}_i$
as a partial function whose inputs are quintuples $(X,Y,\gamma,h_v,h_e)$ where
\begin{itemize}
\item $X\subseteq X_i$,
\item $Y\subseteq E(G[X_i])$,
\item $\gamma: X_i\setminus X\rightarrow \{0,\ldots,k_v+k_e\}$,
\item $h_v\leq k_v$ and
\item $h_e\leq k_e$.
\end{itemize}
The value of ${\bf table}_i$ is a minimum cost pair $(U,D)\in 2^{V(G_i)}\times 2^{E(G_i)}$
with the following properties:
\begin{enumerate}[(i)]
\item for any $v\in U$ and any $e\in D$, $v$ and~$e$ are not incident,
\item $w(U)\leq h_v$ and $w(D)\leq h_e$,
\item $U\cap X_i=X$ and $D\cap E(G[X_i])=Y$,
\item for every $v\in X_i\setminus X$, the number of neighbours of~$v$ in~$G_i$ that belong in $U\setminus X_i$
plus the number of edges of $D\setminus E(G[X_i])$ that are incident to~$v$ is exactly~$\gamma(v)$,
\item for each $v\in V(G_i)\setminus X_i$, $d_{G_i'}(v)=\delta(v)$ where $G_i'=G_i-U-D$,
\end{enumerate}
and, if no such pair $(U,D)$ exists, then ${\bf table}_i(X,Y,d,h_v,h_e)$ is void.

Recall that $X_r=\emptyset$.
Observe that $(G,k_v,k_e,C,\delta,w,c)$ is a yes-instance if and only if
${\bf table}_r(\emptyset,\emptyset,\varnothing,k_v,k_e)$ is non-void (where $\varnothing: \emptyset \rightarrow \{0,\ldots,k_v+k_e\}$). Moreover, in such a case, the value of ${\bf table}_r(\emptyset,\emptyset,\varnothing,k_v,k_e)$ is a minimum-cost solution for this instance.

Now we explain how we construct ${\bf table}_i$ for each $i\in V(T)$.
If~$i$ is a {\sl leaf} node, ${\bf table}_i$ is constructed in a straightforward
way because $X_i=\emptyset$.
Indeed, for $0 \leq h_v \leq k_v$ and $0 \leq h_e \leq k_e$ we set ${\bf table}_i(\emptyset,\emptyset,\varnothing,h_v,h_e)=(\emptyset,\emptyset)$ and have ${\bf table}_i$ void in all other cases.
Hence, it remains to give the construction for {\sl introduce}, {\sl forget}, and {\sl join} nodes. Let $i\in V(T)$ be a node of one of these types. Assume inductively that the function ${\bf table}_{i'}$ for every child~$i'$ of~$i$ has already been constructed.

\newcommand{\improves}{\leftarrowtail}

In what follows we write ${\bf table}_i(X,Y,\gamma,h_v,h_e) \improves (U,D)$
to refer to the following procedure: If ${\bf table}_i(X,Y,\gamma,h_v,h_e)$
is undefined, set it to be equal to $(U,D)$. If ${\bf
table}_i(X,Y,\gamma,h_v,h_e) = (\hat{U},\hat{D})$ and $c(\hat{U}\cup
\hat{D})>c(U\cup D)$, change ${\bf table}_i(X,Y,\gamma,h_v,h_e)$ to be equal
to $(U,D)$. Otherwise, do not change ${\bf table}_i(X,Y,\gamma,h_v,h_e)$.

\medskip
\noindent
{\bf Construction for an introduce node.} Let~$i'$ be the child of~$i$ and $X_i=X_{i'}\cup\{v\}$.
Notice that $N_{G_i}(v)\subseteq X_{i'}$. We start with ${\bf table}_i$ empty.
Then, for each pair $h_v,h_e$ where $h_v\leq k_v$ and $h_e\leq k_e$
and each
pair $((X',Y',\gamma',h_v',h_e'),(U',D'))\in{\bf table}_{i'}$ where $h_v'\leq h_v$ and $h_e'\leq h_e$, we do the following:
\begin{itemize}
\item Let $X\leftarrow X'\cup\{v\}$, $Y\leftarrow Y'$, $\gamma\leftarrow \gamma'$, $U\leftarrow U'\cup \{v\}$, and~$D\leftarrow D'$. \\ If $h_v\geq h_v'+w(v)$, then
${\bf table}_i(X,Y,\gamma,h_v,h_e) \improves (U,D)$.
\item Let $X\leftarrow X'$, $U\leftarrow U'$, $\gamma\leftarrow \gamma'\cup\{(v,0)\}$. \\
For every $L\subseteq \{vu\mid vu\in E(G),u\in X_{i'}\setminus X'\}$, let $Y\leftarrow Y'\cup L$, $D\leftarrow D'\cup L$, and
if $h_e\geq h_e'+w(L)$, then
${\bf table}_i(X,Y,\gamma,h_v,h_e) \improves (U,D)$.\end{itemize}

\medskip
\noindent
{\bf Construction for a forget node.} Let~$i'$ be the child of~$i$ and $X_i=X_{i'}\setminus\{v\}$.
We start with ${\bf table}_i$ empty.
For each pair $((X',Y',\gamma',h_v,h_e),(U,D))\in{\bf table}_{i'}$, we do the following.
\begin{itemize}
\item If $v\in X'$ then let $X\leftarrow X'\setminus\{v\}$, $Y\leftarrow Y'$, and define~$\gamma$ by replacing in~$\gamma'$
each pair $(u,\gamma'(u))$ where $uv\in E(G)$ and $u\in X_i\setminus X$ by the pair $(u,\gamma'(u)+1)$.\\
If $\max_{u \in X_i \setminus X} \gamma(u) \leq k_v+k_e$, then
${\bf table}_i(X,Y,\gamma,h_v,h_e) \improves (U,D)$.
\item If $v\notin X'$, then
let $X\leftarrow X'$, $L\leftarrow \{vu\in E(G)\mid u\in X_i\}\cap Y'$,
$Y\leftarrow Y'\setminus L$, and define~$\gamma$ by replacing in $\gamma^-=\gamma'\setminus \{(v,\gamma'(v))\}$
each pair $(u,\gamma'(u))$ where $uv\in L$ by the pair $(u,\gamma'(u)+1)$.\\
If $\delta(v)=d_G(v)-|L|-\gamma'(v)$ and $\max_{u \in X_i \setminus X} \gamma(u)\leq k_v+k_e$, then
${\bf table}_i(X,Y,\gamma,h_v,h_e) \improves (U,D)$.\end{itemize}

\medskip
\noindent
{\bf Construction for a join node.} Let~$i'$ and~$i''$ be the children of~$i$. We start with ${\bf table}_i$ empty.
For each pair $((X,Y,\gamma',h_v',h_e'),(U',D'))\in{\bf table}_{i'}$
and each pair $((X,Y,\gamma'',h_v'',h_e''),(U'',D'')\in{\bf table}_{i''}$
we do the following.

\begin{itemize}
\item
Let $\gamma\leftarrow \gamma'+\gamma''$, $U\leftarrow U'\cup U''$ and $D\leftarrow D'\cup D''$. \\
If $\max_{u \in X_i \setminus X} \gamma(u) \leq k_v+k_e$, then for any two integers $h_v,h_e$ such that
$h_v'+h_v''-w(X)\leq h_v\leq k_v$ and $h_e'+h_e''-w(Y)\leq h_e\leq k_e$,
${\bf table}_i(X,Y,\gamma,h_v,h_e) \improves (U,D)$.\end{itemize}

\medskip
Using standard arguments, it is straightforward to verify the correctness of the algorithm. To evaluate the running time, recall that ${\bf table}_i$ receives a
quintuple $(X,Y,\gamma,h_v,h_e)$ as input. There are at most~$2^{t+1}$ possible choices
for~$X$, $2^{3(t+1)-6}=2^{3t-3}$ choices of~$Y$ (because of the planarity
of~$G$), $(k_v+k_e+1)^{t+1}$ choices of~$\gamma$, $k_v+1$ possible values
of~$h_v$ and $k_e+1$ possible values for~$h_e$. We therefore
have that each ${\bf table}_i$ has
has $(k_v+k_e)^{O(t)}$ entries. This implies that the running time
of the dynamic programming algorithm is $(k_v+k_e)^{O(t)}\cdot n$.\medskip

Now we consider \textsc{DCPGGD}. The difference is that we have to keep track of components of a partial solution as is standard for dynamic programming algorithms for graphs of bounded treewidth with a connectivity condition such as, e.g. the \textsc{Steiner Tree} problem.
Let $(G,k_v,k_e,C,\delta,w,c)$ be an instance of \textsc{DCPGGD} where $\tw(G)\leq t$
and $\delta(v)\leq d_G(v)\leq \delta(v)+k_v+k_e$ for $v\in V(G)$.
Without loss of generality we assume that a nice tree decomposition $(\mathcal{X},T)$ of~$G$ with treewidth at most~$t$ is given and apply a dynamic programming algorithm over $(\mathcal{X},T)$.
Let $i\in V(T)$.

We define ${\bf table}_i^c$
as a partial function whose inputs are quintuples $({\cal P},Y,\gamma,h_v,h_e)$ where
\begin{itemize}
\item $\mathcal{P}=\{P_0,\ldots,P_s\}$ is a partition of~$X_i$,
\item $Y\subseteq E(G[X_i])$,
\item $\gamma: X_i\setminus X\rightarrow \{0,\ldots,k_v+k_e\}$,
\item $h_v\leq k_v$ and
\item $h_e\leq k_e$.
\end{itemize}
The value of ${\bf table}_i^c$ is a minimum cost pair $(U,D)\in 2^{V(G_i)}\times 2^{E(G_i)}$ with the following properties:
\begin{enumerate}[(i)]
\item for any $v\in U$ and any $e\in D$, $v$ and~$e$ are not incident,
\item $w(U)\leq h_v$ and $w(D)\leq h_e$,
\item $U\cap X_i=P_0$ and $D\cap E(G[X_i])=Y$,
\item for every $v\in X_i\setminus X$, the number of neighbours of~$v$ in~$G_i$ that belong in $U\setminus X_i$,
plus the number of edges of $D\setminus E(G[X_i])$ that are incident to~$v$ is exactly~$\gamma(v)$,
\item for each $v\in V(G_i)\setminus X_i$, $d_{G_i'}(v)=\delta(v)$ where $G_i'=G_i-U-D$,
\item if $s=0$, then $G_i'=G_i-U-D$ is connected and if $s\geq 1$, then~$G_i'$ has~$s$ components $H_1,\ldots,H_s$ such that $V(H_i)\cap X_h=P_i$ for $h\in\{1,\ldots,s\}$,
\end{enumerate}
and, if no such pair $(U,D)$ exists, then ${\bf table}_i^c({\cal P},Y,d,h_v,h_e)$ is void.

As in the non-connected case, $(G,k_v,k_e,C,\delta,w,c)$ is a yes-instance
if and only if
${\bf table}_r^c(\emptyset,\emptyset,\varnothing,k_v,k_e)$ is non-void and the value of ${\bf table}_r^c(\emptyset,\emptyset,\varnothing,k_v,k_e)$, if exists, is a minimum-cost solution for this instance.

The partial function ${\bf table}^c_i$ is constructed for every $i\in V(T)$ similarly to the construction of ${\bf table}_i$ for \textsc{DPGGD}. Because there are at most $(t+1)^{t+1}$ partitions~$\mathcal{P}$ of each~$X_i$, we have that each table contains $(t(k_v+k_e))^{O(t)}$ entries. Therefore, the running time of the dynamic programming algorithm is $(t(k_v+\nobreak k_e))^{O(t)} n$.
\qed
\end{proof}

We are now ready to present our two main results, starting with the one for DPGGD.

\begin{theorem}\label{thm:ker-1}
\textsc{DPGGD} has a polynomial kernel when parameterized by $k_v+k_e$.
\end{theorem}

\begin{proof}
Let $(G,k_v,k_e,C,\delta,w,c)$ be an instance of \textsc{DPGGD}.
By Lemma~\ref{lem:normalization}, we may assume that this instance is normalized.
By Lemma~\ref{lem:ds}, if $(G,k_v,k_e,C,\delta,w,c)$ is a yes-instance, then~$G$ has a 2-dominating set of size at most $k_v+2k_e$. By
Lemma~\ref{lem:ds_pd}, there is a fixed constant~$\alpha$ such that~$G$ has an
$(\alpha(k_v+2k_e),\alpha)$-protrusion decomposition, and such a decomposition,
if it exists, can be constructed in polynomial time.
To simplify later arguments, we may assume $\alpha \geq 3$.
Clearly, if we fail to
obtain such a decomposition, we return a no-answer and stop. Hence, from now on we
assume that an $(\alpha(k_v+2k_e),\alpha)$-protrusion decomposition
$\Pi=\{R_0,\ldots,R_p\}$ of~$G$ is given.
As before, we keep the same notation $\delta,w,c$ for the restrictions of these functions.
Again, we will introduce new reduction rules. We will keep the notation for~$G$ and for the parameters unchanged where this is well-defined.
We also assume that if we consider sets of vertices or edges associated with the considered instance and delete vertices or edges from the graph, then we also delete these elements from the associated sets.

For each $i\in \{1,\ldots,p\}$, we construct $W_i\subseteq R_i$ and $L_i\subseteq E_G(R_i)$.
To do this, we consider the set~${\cal Q}$ of all possible quintuples ${\bf q}=(h_v,h_e,X,Y,\delta')$ such that
\begin{itemize}
\item $0\leq h_v\leq k_v$ and $0 \leq h_e\leq k_e$,
\item $X\subseteq N_G(R_i)$ and $Y\subseteq E(G[N_G(R_i)\setminus X])$, and
\item We define $F=G[R_i^+]- X-Y$ and require that $\delta'\colon
V(F)\rightarrow \mathbb{N}_0$ is a function such that
$\delta'(v)\leq d_F(v)\leq \delta'(v) +k_v+k_e$
for $v\in N_G(R_i)\setminus X$ and
$\delta'(v)=\delta(v)$ for $v\in R_i$.
\end{itemize}
Observe that there are at most~$2^\alpha$ sets~$X$, at most~$2^{3\alpha-6}$ sets~$Y$,
at most $(k_v+\nobreak 1)(k_e+\nobreak 1)$ pairs $h_v,h_e$, and for each~$X$, there are at most $(k_v+k_e+1)^\alpha$ possibilities for~$\delta'$.
Therefore $|{\cal Q}| \leq 2^\alpha 2^{3\alpha-6} (k_v+1)(k_e+1) (k_v+k_e+1)^\alpha =(k_v+k_e)^{O(\alpha)}$.

For each ${\bf q}=(h_v,h_e,X,Y,\delta')\in{\cal Q}$, we construct an instance
$I_{\bf q}=(F,h_v,h_e,C,\delta',w',c)$ of \textsc{DPGGD}
such that
\begin{itemize}
\item $w'(v)=k_v+1$ for $v\in N_G(R_i)\setminus X$ and $w'(v)=w(v)$ for $v\in R_i$ and
\item $w'(e)=k_e+1$ for $e\in E(G[N_G(R_i)\setminus X])\setminus Y$ and $w'(e)=w(e)$ for all other edges of~$F$.
\end{itemize}
By Lemma~\ref{lem:tw}, we can solve the problem for this instance in
$(k_v+k_e)^{O(\alpha)}$ time.
Let $(U_{\bf q},D_{\bf q})$ denote the obtained solution of minimum cost and set $U_{\bf q}=D_{\bf q}=\emptyset$ if no solution exists for~$I_{\bf q}$.
Let
$$W_i=\bigcup_{{\bf q}\in {\cal Q}} U_{\bf q}\mbox{\  and \ } L_i=\bigcup_{{\bf q}\in {\cal Q}} D_{\bf q}.$$
Because each~$U_{\bf q}$ has at most~$k_v$ vertices and each~$D_{\bf q}$ has at most~$k_e$ edges, we obtain that
$|W_i|\leq |{\cal Q}|k_v \leq (k_v+1)(k_e+1)\cdot 2^\alpha\cdot2^{3\alpha-6}\cdot (k_v+k_e+1)^\alpha\cdot k_v$
and $|L_i|\leq |{\cal Q}|k_e \leq (k_v+1)(k_e+1)\cdot 2^\alpha\cdot2^{3\alpha-6}\cdot (k_v+k_e+1)^\alpha\cdot k_e$. Hence, the size of~$W_i$ and~$L_i$ is $(k_v+k_e)^{O(\alpha)}$.

Let $W=R_0\cup\bigcup_{i\in\{1,\ldots,p\}} W_i$ and $L=E(G[R_0])\cup \bigcup_{i\in \{1,\ldots,p\}} L_i$. Because $\max\{p,|R_0|\}\leq \alpha (k_v+2k_e)$, we have that
$|W|=(k_v+k_e)^{O(\alpha)}$ and $|L|=(k_v+\nobreak k_e)^{O(\alpha)}$. We prove the following claim.

\clm{\label{clm:first} If $(G,k_v,k_e,C,\delta,w,c)$ is a yes-instance of \textsc{DPGGD}, then it has an efficient solution $(U,D)$ of minimum cost such that
$U\subseteq W$ and $D\subseteq L$.}

\medskip
\noindent
We prove Claim~\ref{clm:first} as follows.
Let $(U,D)$ be an efficient solution for $(G,k_v,k_e,C,\delta,w,c)$ of minimum cost such that $s=|U\setminus W|+|D\setminus L|$ is minimum. If $s=0$, then the claim is fulfilled. Suppose, for contradiction, that $s>0$. This means that there is an $i\in\{1,\ldots,p\}$ such that
$(U \cap R_i) \setminus W_i \neq \emptyset$ or $(D \cap E_G(R_i)) \setminus L_i \neq \emptyset$.
Let $X=U\cap N_G(R_i)$, $Y=D\cap E(N_G(R_i))$ and $F=G[R_i^+]-X-Y$. Let $h_v=|U\cap V(F)|$ and $h_e=|D\cap E(F)|$.
For each vertex $v\in N_G(R_i)\setminus X$, let~$d_v$ be the total number of vertices in $U\setminus V(F)$ adjacent to~$v$ plus the number of edges in $D\setminus E(F)$ incident to~$v$.
Let $\delta'(v)=d_F(v)-(d_G(v)-\delta(v)-d_v)$
for $v\in N_G(R_i)\setminus X$ and $\delta'(v)=\delta(v)$ for all other vertices of~$F$.

Clearly, $(F,h_v,h_e,C,\delta',w',c)=I_{\bf q}$ is the instance of \textsc{DPGGD} when ${\bf q}=(h_v,h_e,X,Y,\delta')$
if we set~$w'$ as before.
Let $U'=U\cap V(F)$ and $D'=D\cap E(F)$.
Then $(U',D')$ is a solution for the instance~$I_{\bf q}$ and, therefore~$I_{\bf q}$ is a yes-instance.
In particular, this means that there is a solution $(U'',D'')$ for $I_{\bf q}=(F,h_v,h_e,C,\delta',w',c)$ that was constructed by the aforementioned procedure for the construction of~$W_i$ and~$L_i$.
Clearly, $U''\subseteq W_i\subseteq W$ and $D''\subseteq L_i\subseteq L$.
Because our algorithm for graphs of bounded treewidth finds a solution of minimum cost, it follows that $c(U''\cup D'')\leq c(U'\cup D')$.
It remains to observe that $(\hat{U},\hat{D})$, where $\hat{U}=(U\setminus U')\cup U''$ and $\hat{D}=(D\setminus D')\cup D''$, is a solution for $(G,k_v,k_e,C,\delta,w,c)$ with $c(\hat{U} \cup \hat{D}) \leq c(U \cup D)$, but this contradicts the choice of $(U,D)$ because $|\hat{U}\setminus W|+|\hat{D}\setminus L|<s$.
This completes the proof of Claim~\ref{clm:first}.

\medskip
\noindent
Let $S=\{v\in V(G) \mid d_G(v)=\delta(v)\}\setminus W$ and $T=\{v\in V(G)\mid d_G(v)>\delta(v)\}\setminus W$; because the instance we consider is normalized, these sets form a partition of $V(G)\setminus W$ (note that these sets may be empty).
If $v\in S$, then for any efficient solution $(U,D)$ such that
$U\subseteq W$ and $D\subseteq L$, $v$ is not adjacent to any vertex of~$U$ and not incident to any edge of~$L$.
This implies that it is safe to exhaustively apply the following rule without destroying the statement of Claim~\ref{clm:first}.

\noindent
\begin{quote}
{\bf Set adjustment rule.} If there is a vertex $v\in S$ that is adjacent to a vertex $u\in W$, then set $W=W\setminus\{u\}$ and set $S=S\cup\{u\}$ if $d_G(u)=\delta(u)$ and
set $T=T\cup\{u\}$ if $d_G(u)>\delta(u)$.
If $v \in S$, remove any edge incident to~$v$ from~$L$.
\end{quote}

\noindent
By Claim~\ref{clm:first}, it is safe to modify the weights as follows.

\noindent
\begin{quote}
{\bf Weight adjustment rule.} Set $w(v)=k_v+1$ for $v\in V(G)\setminus W$ and set $w(e)=k_e+1$ for $e\in E(G)\setminus L$.
\end{quote}

\noindent
After the exhaustive application of the {\bf set adjustment rule}, we have that $N_G(S)\subseteq T$.
Now it is safe to remove~$S$.

\noindent
\begin{quote}
{\bf $S$-reduction rule.} If $v\in S$, then remove~$v$ and set $\delta(u)=\delta(u)-1$ for $u\in N_G(v)$.
If $\delta(u)<0$ for some $u\in N_G(v)$, then return a no-answer and stop.
\end{quote}

To show that the above rule is safe, let $G'=G-S$ and let~$\delta'$ be the function obtained from~$\delta$ by the application of the rule. Suppose that $(G,k_v,k_e,C,\delta,w,c)$ is a yes-instance. Then, by Claim~\ref{clm:first}, we have a solution $(U,D)$ such that $U\subseteq W$ and $D\subseteq L$. Because $N_G(S)\subseteq T$, $T\cap W=\emptyset$ and the vertices of~$S$ are not incident to edges of~$L$, it follows that we do not stop and
$(U,D)$ is a solution for $(G',k_v,k_e,C,\delta',w,c)$. Now let $(U,D)$ be a solution for $(G',k_v,k_e,C,\delta',w,c)$. Because of the application of the {\bf weight adjustment rule},
$U\subseteq W$ and $D\subseteq L$. Because $N_G(S)\subseteq T$, $T\cap W=\emptyset$ and the vertices of~$S$ are not incident to edges of~$L$, we have that $(U,D)$ is a solution for $(G,k_v,k_e,C,\delta,w,c)$. This completes the proof that the {\bf $S$-reduction rule} is safe.

Let $W'=W\cup V(L)$ and $T'=T\setminus V(L)$. Clearly,
$|W'|\leq |W|+2|L|= (k_v+k_e)^{O(\alpha)}$.

Using similar arguments to those for the {\bf $S$-reduction rule}, the following rule is also safe.

\noindent
\begin{quote}
{\bf $T'$-reduction rule.} If $uv\in E(G[T'])$, then remove~$uv$ and set $\delta(u)=\delta(u)-1$ and $\delta(v)=\delta(v)-1$. If $\delta(u)<0$ or $\delta(v)< 0$, then return a no-answer and stop.
\end{quote}

After the exhaustive application of the above rule,~$T'$ is an independent set in the obtained graph~$G$. Some of the vertices of this independent set may have the same neighbourhoods. We deal with them using the next rule.

\noindent
\begin{quote}
{\bf Twin reduction rule.} Suppose there are $u,v\in T'$ with $N_G(u)= N_G(v)$. If $\delta(u)=\delta(v)$, then remove~$v$ and set $\delta(x)=\max\{0,\delta(x)-1\}$ for $x\in N_G(u)$.
If $\delta(u) \neq \delta(v)$ then return a no-answer and stop.
\end{quote}

To prove that the above rule is safe, consider a pair of vertices $u,v\in T'$ with $N_G(u)=N_G(v)$ and $\delta(u)=\delta(v)$. Let $G'=G-v$ and let~$\delta'$ denote the function obtained from~$\delta$ by the rule.
Suppose that $(G,k_v,k_e,C,\delta,w,c)$ is a yes-instance. Then we have a solution $(U,D)$ such that $U\subseteq W$ and $D\subseteq L$.
Notice that $T'\cap U=\emptyset$ and the vertices of~$T'$ are not incident to the edges of~$L$.
Note that $u,v\notin U$ and if $x\in N_G(u)$ then $ux,vx\notin D$.
We have that~$U$ contains exactly $d_G(u)-\delta(u)$ vertices that are adjacent to~$u$.
Therefore, $(U,D)$ is a solution for $(G',k_v,k_e,C,\delta',w,c)$. Now assume that $(U,D)$ is a solution for $(G',k_v,k_e,C,\delta',w,c)$. By the same arguments,~$U$ contains exactly $d_{G'}(u)-\delta'(u)$ vertices that are adjacent to~$u$.
Also if $x\in N_G(u)$ and $\delta'(x)=0$, then $x\in U$, because $u\notin U$ and $ux\notin D$.
Because $N_G(u)=N_G(v)$, $\delta(u)=\delta(v)$ and~$T'$ is an independent set,~$U$ contains $d_{G}(u)-\delta(u)$ vertices that are adjacent to~$u$ and
$d_{G}(v)-\delta(v)$ vertices that are adjacent to~$v$.
It follows that $(U,D)$ is a solution for $(G,k_v,k_e,C,\delta,w,c)$.
Now consider the case when $N_G(u)=N_G(v)$ and $\delta(u) \neq \delta(v)$. Suppose, for contradiction that there is a solution $(U,D)$. By the above arguments,~$U$ contains exactly $d_{G}(u)-\delta(u)$ vertices that are adjacent to~$u$ and $d_{G}(v)-\delta(v)$ vertices that are adjacent to~$v$. Since $N_G(u)=N_G(v)$ and $\delta(u) \neq \delta(v)$, this is a contradiction, so there cannot be such a solution.

\medskip
After the exhaustive application of the above rule for any two vertices $u,v\in\nobreak T'$, we have that $N_G(u)\neq N_G(v)$.
Let $T_{0}',T_1',T_2',T_{\geq 3}'$ denote the sets of vertices in~$T'$ that are of degree $0, 1, 2$ and at least~$3$ respectively. Observe that $d_G(v)>\delta(v)\geq 0$ for $v\in T'$.
Therefore, $T_{0}'=\emptyset$ and $T_1',T_2',T_{\geq 3}'$ form a partition of~$T'$ (note that these sets may be empty).
By the {\bf twin reduction rule}
$|T_1'|=|N_G(T_1')|\leq |W'|$ and $|T_2'|\leq \binom{|N_G(T_2')|}{2}\leq \frac{1}{2}|W'|(|W'|-1)$.
By Lemma~\ref{lem:bound-bip}, $|T_{\geq 3}'|\leq 2|N_G(T')|-4\leq 2|W'|-4$ (or $|T_{\geq 3}'|=0$).
We have that $|V(G)|=|W'|+|T'|=|W'|+|T_1'|+|T_2'|+|T_{\geq 3}'|\leq \frac{1}{2}|W'|^2+\frac{7}{2}|W'|$. Since,~$W'$ has $(k_v+k_e)^{O(\alpha)}$ vertices, we obtain that the obtained graph~$G$ has size~$k^{O(1)}$
where $k=k_v+k_e$, i.e. we have a polynomial kernel for \textsc{DPGGD}.

To complete the proof, it remains to observe that the construction of the normalized instance can be done in polynomial time by Lemma~\ref{lem:normalization}, the construction of~$W$ and~$L$ can be done in polynomial time by Lemma~\ref{lem:tw}, and all the subsequent reduction rules can be applied in polynomial time.
\qed
\end{proof}

The proof of our second main result is based on the same approach as the proof of Theorem~\ref{thm:ker-1}, but it is more technically involved because we have to ensure connectivity of the graph obtained by the editing.

\begin{theorem}\label{thm:ker-2}
\textsc{DCPGGD} has a polynomial kernel when parameterized by $k_v+k_e$.
\end{theorem}

\begin{proof}
Let $(G,k_v,k_e,C,\delta,w,c)$ be an instance of \textsc{DCPGGD}.
By Lemma~\ref{lem:normalization}, we may assume that this instance is normalized.
By Lemma~\ref{lem:ds}, if $(G,k_v,k_e,C,\delta,w,c)$ is a yes-instance, then~$G$ has a 2-dominating set of size at most $k_v+2k_e$. By
Lemma~\ref{lem:ds_pd}, there is a fixed constant~$\alpha$ such that~$G$ has an
$(\alpha(k_v+2k_e),\alpha)$-protrusion decomposition, and such a decomposition,
if it exists, can be constructed in polynomial time.
To simplify later arguments, we may assume $\alpha \geq 3$.
Clearly, if we fail to
obtain such a decomposition, we return a no-answer and stop. Hence, from now on we
assume that an $(\alpha(k_v+2k_e),\alpha)$-protrusion decomposition
$\Pi=\{R_0,\ldots,R_p\}$ of~$G$ is given.
As before, we keep the same notation $\delta,w,c$ for the restrictions of these functions.
Again, we will introduce new reduction rules. We will keep the notation for~$G$ and for the parameters unchanged where this is well-defined.
We also assume that if we consider sets of vertices or edges associated with the considered instance and delete vertices or edges from the graph, then we also delete these elements from the associated sets.

For each $i\in \{1,\ldots,p\}$, we construct $W_i\subseteq R_i$ and $L_i\subseteq E_G(R_i)$.
To do this, we consider the set~${\cal Q}$ of all possible sextuples ${\bf q}=(h_v,h_e,X,Y,{\cal P},\delta')$ such that
\begin{itemize}
\item $0\leq h_v\leq k_v$ and $0 \leq h_e\leq k_e$,
\item $X\subseteq N_G(R_i)$ and $Y\subseteq E(G[N_G(R_i)\setminus X])$,
\item ${\cal P}=\{P_1,\ldots,P_s\}$ is a set covering of $N_G(R_i)\setminus X$, with $s \leq |N_G(R_i)\setminus X|$,
\item We define $F=G[R_i^+]- X-Y$ and require that $\delta'\colon
V(F)\rightarrow \mathbb{N}_0$ is a function such that
$\delta'(v)\leq d_F(v)\leq \delta'(v)+k_v+k_e$
for $v\in N_G(R_i)\setminus X$ and
$\delta'(v)=\delta(v)$ for $v\in R_i$.
\end{itemize}

Observe that there are at most~$2^\alpha$ sets~$X$, at most~$2^{3\alpha-6}$ sets~$Y$,
at most $(k_v+\nobreak 1)(k_e+\nobreak 1)$ pairs $h_v,h_e$, and for each~$X$, there are at most~$2^{\alpha^2}$ possible set covers~$\mathcal{P}$ and at most $(k_v+k_e+1)^\alpha$ possibilities for~$\delta'$.
Therefore $|{\cal Q}| \leq 2^\alpha 2^{3\alpha-6} (k_v+\nobreak 1)\allowbreak (k_e+1) 2^{\alpha^2} (k_v+k_e+1)^\alpha =(k_v+k_e)^{O(\alpha^2)}$.

For each ${\bf q}=(h_v,h_e,X,Y,{\cal P},\delta')\in {\cal Q}$, we construct an instance
$I_{\bf q}=(F_{\cal P},h_v,h_e,C,\delta'',w',c')$ of \textsc{DCPGGD}
such that
\begin{itemize}
\item $F_{\cal P}$ is the graph obtained
from~$F$ by adding a set of~$s$ new vertices $Z=\{z_1,\ldots,z_s\}$ and
making~$z_i$ adjacent to all the vertices of~$P_i$. If ${\cal P}=\emptyset$,
which means that $N_G(R_i^+)= X$, then we simply have that $Z=\emptyset$ and $F_{\cal P}=F$.
\item $\delta''(v)=d_{F_{\cal P}}(v)$ for $v \in Z$ and $\delta''(v)=\delta'(v)$ for $v \in V(F_{\cal P})\setminus Z$.
\item $w'(v)=k_v+1$ for $v\in (N_G(R_i)\setminus X)\cup Z$ and $w'(v)=w(v)$ for $v\in R_i$.
\item $w'(e)=k_e+1$ for $e\in (E(G[N_G(R_i)\setminus X])\setminus Y)\cup E_{F_{\cal P}}(Z)$, and $w'(e)=w(e)$ for all other edges of~$F_{\cal P}$.
\item $c'(v)=0$ for $v\in Z$ and $c'(v)=c(v)$ for $v \in V(F_{\cal P})\setminus Z$; $c'(e)=0$ for $e\in E_{F_{\cal P}}(Z)$ and $c'(e)=c(e)$ for all other edges in~$F_{\cal P}$.
\end{itemize}

Since $|Z| \leq |N_G(R_i)| \leq \alpha$, it follows that $|Z|\leq \alpha$ and therefore $\tw(F_P)\leq\tw(F)+\alpha\leq 2\alpha$.
We can check in linear time whether~$F_{\cal P}$ is planar~\cite{HT74}. If it is not, then~$I_{\bf q}$ is not a valid instance of \textsc{DCPGGD} and we set $(U_{\bf q},D_{\bf q})=(\emptyset,\emptyset)$.
Otherwise, by Lemma~\ref{lem:tw}, we can solve \textsc{DCPGGD} for~$I_{\bf q}$ in
$(\alpha(k_v+k_e))^{O(\alpha)}\cdot \poly(n)$ 
time and find a solution of minimum cost.
Let $(U_{\bf q},D_{\bf q})$ be the obtained solution of minimum cost and let $U_{\bf q}=D_{\bf q}=\emptyset$ if no solution exists.
Notice that $Z\cap U_{\bf q}=\emptyset$, because the vertices of~$Z$ have weight $k_v+1$, and~$D_{\bf q}$ has no edges incident to the vertices of~$Z$, because these edges have weight $k_e+1$.
Let
$$W_i=\bigcup_{{\bf q}\in {\cal Q}} U_{\bf q} \mbox{\  and \ }L_i= \bigcup_{{\bf q}\in {\cal Q}} D_{\bf q}.$$
Because each~$U_{\bf q}$ has at most~$k_v$ vertices and each~$D_{\bf q}$ has at most~$k_e$ edges, we obtain that
$|W_i|\leq |{\cal Q}|k_v \leq (k_v+1)(k_e+1)\cdot 2^\alpha\cdot2^{3\alpha-6}\cdot 2^{\alpha^2} \cdot (k_v+k_e+1)^\alpha\cdot k_v$
and $|L_i|\leq |{\cal Q}|k_e \leq (k_v+1)(k_e+1)\cdot 2^\alpha\cdot2^{3\alpha-6}\cdot 2^{\alpha^2}\cdot (k_v+k_e+1)^\alpha\cdot k_e$. Hence, the size of~$W_i$ and~$L_i$ is $(k_v+k_e)^{O(\alpha^2)}$.

Let $W=R_0\cup\bigcup_{i\in\{1,\ldots,p\}} W_i$ and $L=E(G[R_0])\cup \bigcup_{i\in \{1,\ldots,p\}} L_i$. Because $\max\{p,|R_0|\}\leq \alpha (k_v+2k_e)$, we have that
$|W|=(k_v+k_e)^{O(\alpha^2)}$ and $|L| = (k_v+\nobreak k_e)^{O(\alpha^2)}$. We prove the following claim.

\clm{\label{clm:second} If $(G,k_v,k_e,C,\delta,w,c)$ is a yes-instance of \textsc{DCPGGD}, then it has an efficient solution $(U,D)$ of minimum cost such that
$U\subseteq W$ and $D\subseteq L$.}

\medskip
\noindent
We prove Claim~\ref{clm:second} as follows.
Let $(U,D)$ be an efficient solution for $(G,k_v,k_e,C,\delta,w,c)$ of minimum cost such that $s=|U\setminus W|+|D\setminus L|$ is minimum. If $s=0$, then the claim is fulfilled. Suppose, for contradiction, that $s>0$. This means that there is an $i\in\{1,\ldots,p\}$ such that
$(U \cap R_i) \setminus W_i \neq \emptyset$ or $(D \cap E_G(R_i)) \setminus L_i \neq \emptyset$.

Let $X=U\cap N_G(R_i)$, $Y=D\cap E(N_G(R_i))$ and $F=G[R_i^+]-X-Y$. Let $h_v=|U\cap V(F)|$ and $h_e=|D\cap E(F)|$.
If $X\neq N_G(R_i)$, then consider the graph $H=G-U-D-R_i^+$ and let $H_1,\ldots,H_s$ denote the components of~$H$.
Next, starting with the graph $H'=G-U-D-R_i$, contract each~$H_j$ to a single vertex~$z_j$ and call the resulting graph~$H''$.
Note that $Z=\{z_1,\ldots,z_s\}$ is an independent set in~$H''$.
By the definition of protrusion decomposition, every vertex of $N_G(R_i) \setminus X$ is adjacent to at least one vertex in~$Z$.
Likewise, since $G-U-D$ is connected, every vertex~$z_j$ must have a neighbour in $N_G(R_i)\setminus X$.
If there is a vertex $z_j\in Z$ such that removing it from~$H''$ does not increase the number of components in~$H''$ and every vertex in $N_G(R_i)\setminus X$ has a neighbour in $Z \setminus \{z_j\}$ then we remove~$z_j$ from~$H''$ and from~$Z$.
Doing this exhaustively, we obtain a graph with $|Z| \leq |N_G(R_i) \setminus X| \leq \alpha$. Call this graph~$F_{\cal P}$.
Without loss of generality assume $Z=\{z_1,\ldots,z_t\}$.
Let $P_j=N_{H''}(z_j)$ for $j \in \{1,\ldots,t\}$.
Then $\mathcal{P}=\{P_1,\ldots,P_t\}$ is a set cover of $N_G(R_i)\setminus X$ containing at most~$\alpha$ sets.
If $X=N_G(R_i)$, then set ${\cal P}=\emptyset$ and $F_{\cal P}=F$.
Now~$F_{\cal P}$ is precisely the graph constructed from~$F$ and~$\mathcal{P}$ earlier.
Note that~$F_{\cal P}$ is planar since it is obtained from~$G$ by contractions, vertex deletions and edge deletions.

For each vertex $v\in N_G(R_i)\setminus X$, let~$d_v$ be the total number of vertices in $U\setminus V(F)$ adjacent to~$v$ plus the number of edges in $D\setminus E(F)$ incident to~$v$.

Let $\delta'(v)=d_{F_{\cal P}}(v)-(d_G(v)-\delta(v)-d_v)$
for $v\in N_G(R_i)\setminus X$ and $\delta'(v)=\delta(v)$ for other vertices of~$F_{\cal P}$.
Set $w', c'$ and~$\delta''$ as before.

Clearly, $I_{\bf q}=(F_{\cal P},h_v,h_e,C,\delta',w',c')$ is an
instance of \textsc{DCPGGD} when ${\bf q}=(h_v,h_e,X,Y,{\cal P},\delta')$.
Let $U'=U\cap V(F)$ and $D'=D\cap E(F)$.
Then $(U',D')$ is a solution for the instance~$I_{\bf q}$ and, therefore~$I_{\bf q}$ is a yes-instance.
 
In particular, this means that there is a solution $(U'',D'')$ for $I_{\bf q}=(F_{\cal P},h_v,h_e,C,\delta'',w',c')$ that was constructed by the aforementioned procedure for the construction of~$W_i$ and~$L_i$.
Clearly, $U''\subseteq W_i\subseteq W$ and $D''\subseteq L_i\subseteq L$.
Because our algorithm for graphs of bounded treewidth finds a solution of minimum cost, it follows that $c(U''\cup D'')\leq c(U'\cup D')$.
It remains to observe that $(\hat{U},\hat{D})$, where  $\hat{U}=(U\setminus U')\cup U''$ and $\hat{D}=(D\setminus D')\cup D''$, is a solution for $(G,k_v,k_e,C,\delta,w,c)$ with $c(\hat{U} \cup \hat{D}) \leq c(U \cup D)$, but this contradicts the choice of $(U,D)$ because $|\hat{U}\setminus W|+|\hat{D}\setminus L|<s$.
This completes the proof of Claim~\ref{clm:second}.

\medskip
\noindent
If $v\in \overline{W}=V(G)\setminus W$ and $d_G(v)=\delta(v)$, then for any efficient solution $(U,D)$ such that
$U\subseteq W$ and $D\subseteq L$, $v$ is not adjacent to a vertex of~$U$.
Moreover, $E_G(v)\cap D=\emptyset$, by Observation~\ref{obs:norm-2}.
This implies that it is safe to apply the following rule without destroying the statement of Claim~\ref{clm:second}.

\noindent
\begin{quote}
{\bf Set adjustment rule.} If there is a vertex $v\in \overline{W}$ with $d_G(v)=\delta(v)$, then set $W=W\setminus N_G(v)$ and set $L=L\setminus E_G(v)$.
\end{quote}

The sets~$W$ and~$L$ give us the following possibility to remove some vertices when there is the unique possibility to satisfy degree restrictions.

\begin{quote}
\noindent
{\bf Vertex deletion rule.} If there is a vertex $v\in \overline{W}$ with $d_G(v)>\delta(v)$ such that $E_G(v)\cap L=\emptyset$ then
\begin{itemize}
\item if $|N_G(v)\cap W|<d_G(v)-\delta(v)$, then return a no-answer and stop;
\item if $|N_G(v)\cap W|=d_G(v)-\delta(v)$, then delete the vertices of $N_G(v)\cap W$ and set $k_v=k_v-w(N_G(v)\cap W)$ and $C=C-c(N_G(v)\cap W)$; if $k_v<0$ or $C<0$, then return a no-answer and stop.
\end{itemize}
\end{quote}

We exhaustively apply the above
two rules
until they can no longer be further applied.
Let $S=\{v\in V(G) \mid d_G(v)=\delta(v)\}\setminus W$. Notice that $N_G(S)\subseteq \overline{W}$ by the {\bf set adjustment rule}. It is easy to see that the following rule is safe.

\begin{quote}
\noindent
{\bf $S$-neighbour rule.} If~$v$ has~$k$ neighbours in~$S$, and $\delta(v)<k$ then return a no-answer and stop.
\end{quote}

We apply the {\bf $S$-neighbour rule} exhaustively. Next, we contract the edges of~$G[S]$.

\begin{quote}
\noindent
{\bf $S$-contraction rule 1.} If~$G$ has two adjacent vertices $u,v\in S$, then we do 
as follows.
\begin{itemize}
\item For any vertex $x\in V(G)\setminus\{u,v\}$ such that $xu,xv\in E(G)$, set
$\delta(x)=\delta(x)-1$.
\item Contract~$uv$; let~$z$ denote the vertex obtained from~$u$ and~$v$.
\item Set $w(z)=k_v+1$ and $c(z)=0$.
\item For $e\in E_G(z)$, set $w(e)=k_e+1$, $c(e)=0$.
\end{itemize}
\end{quote}

We now show that the {\bf $S$-contraction rule 1} is safe. To do this, let $(G',k_v,k_e,C,\delta',w',c')$ denote the instance obtained by an application of the rule.
Let $(U,D)$ be an efficient solution for $(G,k_v,k_e,C,\delta,w,c)$ such that $U \subseteq W$ and $D \subseteq L$. By Observation~\ref{obs:norm-2}, $D$ has no edges incident to~$u$ or~$v$. Also $u,v\notin U$, because $u,v\in S$.
Notice that $\delta(x) \geq 2$ by the {\bf $S$-neighbour rule}.
If $(U',D')$ is an efficient solution for $(G',k_v,k_e,C,\delta',w',c')$, then~$D'$ has no edges~$e$ incident to~$z$, because $w'(e)>k_e$. Similarly, $z\notin U'$ because $w'(z)>k_v$.
Also note that $\delta'(x) \geq 1$ because of the {\bf $S$-neighbour rule}.
We obtain that $(U',D')$ is a solution for the original instance.

We exhaustively apply {\bf $S$-contraction rule 1}. Note that~$S$ is an independent set
 in the obtained instance.

\begin{quote}
\noindent
{\bf Stopping rule.} If~$G$ has two components that contain vertices of~$\overline{W}$, then return a no-answer and stop. Suppose~$\overline{W}$ contains a vertex~$v$ which is isolated in~$G$. In this case if $w(V(G)\setminus \{v\})\leq k_v$ and $c(V(G)\setminus \{v\})\leq C$, then return a $(V(G)\setminus \{v\},\emptyset)$ as a solution and stop, otherwise, return a no-answer and stop.
\end{quote}

Clearly, if~$G$ has two components that contain vertices of~$\overline{W}$, then one of these components should be deleted. By Claim~\ref{clm:second}, we know that if there is a solution then there must be a minimal cost solution that does not delete any vertices of~$\overline{W}$. This contradiction means that there is no solution.
If $v\in\overline{W}$ is an isolated vertex of~$G$, then because $d_G(v)\geq \delta(v)$, it follows that $\delta(v)=d_G(v)$ and we conclude that $(V(G)\setminus\{v\},\emptyset)$ must be a solution. Therefore, the {\bf stopping rule} is safe.

Assume that we do not stop at this stage. Then we obtain the instance $(G,k_v,k_e,C,\delta,w,c)$
of the problem and sets $W,L$ such that
the sets $S=\{v\in V(G) \mid d_G(v)=\delta(v)\}\setminus W$ and $T=\{v\in V(G)\mid d_G(v)>\delta(v)\}\setminus W$
form a partition of~$\overline{W}$ (note that these sets may be empty),~$S$ is an independent set, no vertex of~$\overline{W}$ is isolated in~$G$, and $L\cap E(S)=\emptyset$.
Also for any $v\in S$, $N_G(v)\subseteq T$, by the {\bf set adjustment rule}.

By Claim~\ref{clm:second}, it is safe to modify the weights as follows.

\begin{quote}
\noindent
{\bf Weight adjustment rule.} Set $w(v)=k_v+1$ for $v\in V(G)\setminus W$ and set $w(e)=k_e+1$ for $e\in E(G)\setminus L$.\\
\end{quote}

Our next aim is to bound the size of~$S$. In the proof of Theorem~\ref{thm:ker-1} we simply deleted the vertices of~$S$ and adjusted~$\delta$ appropriately. Here we need to preserve connectivity. Hence, we delete vertices only if this does not destroy connectivity and we use contractions otherwise.

\begin{quote}
\noindent
{\bf $S$-deletion rule.} If, for a vertex $v\in S$, one of the following is fulfilled
\begin{itemize}
\item $d_G(v)=1$,
\item $d_G(v)=2$ and for $\{x,y\}=N_G(v)$, $xy\in E(G)\setminus L$ or
\item there is a vertex $u\in S$ such that $u\neq v$ and $N_G(v)\subseteq N_G(u)$,
\end{itemize}
then delete~$v$ and set
$\delta(x)=\delta(x)-1$
for $x\in N_G(v)$; if $\delta(x)<0$, then return a no-answer and stop.
\end{quote}

\begin{quote}
\noindent
{\bf $S$-contraction rule~2.} If $v\in S$, then let $u\in N_G(v)$ and let $\Delta=d_G(u)-\delta(u)$. For every $vx\in E(v)\setminus\{vu\}$ such that $ux\in E(G)$,
delete~$vx$, add a vertex~$z$ adjacent to~$v$ and~$x$, set $\delta(z)=2$, $w(z)=k_v+1$, $w(zx)=w(zv)=k_e+1$ and
$c(z)=c(zx)=c(zv)=0$ and add~$z$ to~$S$.
Then contract~$uv$ in the obtained graph and set $\delta(y)=d_G(y)+\Delta$,
$w(y)=k_v+1$ and $c(y)=0$ for the vertex~$y$ obtained from~$u$ and~$v$.
\end{quote}

The above two rules are safe, because $N_G(v)\subseteq T$ and the vertices of~$T$ are not included in any solutions.

We apply these rules exhaustively. First we apply the {\bf $S$-deletion rule} whenever it is possible. Then we apply the
{\bf $S$-contraction rule 2}. Notice that the {\bf $S$-contraction rule} creates new vertices that are obtained by subdividing the edges of~$E(v)$ and they are placed in~$S$.
Therefore, it may happen that we can again apply
the {\bf $S$-deletion rule}, and in this case we do so. Finally, we get the graph~$G$ with the following properties:
\begin{enumerate}[(i)]
\item for any $v\in S$, $d_G(v)=2$ and for $\{x,y\}=N_G(v)$, $x,y\in V(L)$, and
\item for any distinct $u,v\in S$, $N_G(u)\neq N_G(v)$ (by the {\bf $S$-deletion rule}). In particular, this means that $|S|\leq (2|L|)^2$.
\end{enumerate}

Let $W'=W\cup V(L)\cup S$ and $T'=T\setminus V(L)$. Clearly,~$W'$ and~$T'$ form a partition of~$V(G)$ (one of the sets could be empty).
Notice that
$|W'|\leq |W|+3|L|=(k_v+k_e)^{O(\alpha^2)}$. Now our aim is to bound the size of~$T'$.

\begin{quote}
\noindent
{\bf $T'$-deletion rule.} If there are two distinct $u,v\in T'$ such that $N_G(u)\cap W'=N_G(v)\cap W'$, $d_G(u)-\delta(u)=d_G(v)-\delta(v)$ and~$v$ is an isolated vertex of~$G[T']$, then
delete~$v$ and set $\delta(x)=\max\{0,\delta(x)-1\}$ for $x\in N_G(u)$.
\end{quote}

To see that the {\bf $T'$-deletion rule} is safe, it is sufficient to recall that $\delta(v)\neq d_G(v)$ because we already applied the {\bf vertex deletion rule}.
Hence, $|N_G(v) \cap W| > d_G(v) - \delta(v)$, so in any solution~$u$ and~$v$ have common adjacent vertices that are not deleted.
Because $E(v)\cap L=\emptyset$ and $E(u)\cap L=\emptyset$, the edges of~$E(u)$ and~$E(v)$ cannot be deleted. Therefore, we maintain connectivity by the {\bf $T'$-deletion rule}.
It is straightforward to verify that the {\bf $T'$-deletion rule} is safe with respect to degree restrictions.

\begin{quote}
\noindent
{\bf $T'$-contraction rule.} If there are two distinct $u,v\in T'$ such that $N_G(u)\cap W'=N_G(v)\cap W'$, $d_G(u)-\delta(u)=d_G(v)-\delta(v)$ and~$u$ and~$v$ are in the same component of~$G[T']$, do the following.
\begin{itemize}
\item For each $vx\in E(v)$ such that $x\notin T'$, delete~$vx$ and set $\delta(x)=\max\{0,\delta(x)-1\}$.
\item Let $y\in N_G(v)$ in the obtained graph and let $\Delta=d_G(y)-\delta(y)$. For every $x\in N_G(v)\cap N_G(y)$, set $\delta(x)=\max\{0,\delta(x)-1\}$.
Contract~$yv$ to a vertex~$z$ and set $\delta(z)=d_G(z)+\Delta$, $w(z)=k_v+1$, $c(z)=0$ and let $w(zx)=k_e+1$, $c(zx)=0$ for every $x\in N_G(z)$.
Add~$z$ to~$T'$.
\end{itemize}
\end{quote}

To show that the {\bf $T'$-contraction rule} is safe, again recall that $\delta(v)\neq |N_G(v)\cap W'|$ because we already applied the {\bf vertex deletion rule}.
Hence, in any solution,~$u$ and~$v$ have common adjacent vertices in~$W'$ that are not deleted.
Because $E(T')\cap L=\emptyset$, the edges of~$E(T')$ cannot be deleted. Therefore, we do not destroy connectivity by the {\bf $T'$-contraction rule}.
It is straightforward to verify that the {\bf $T'$-contraction rule} is safe with respect to degree restrictions.

We exhaustively apply the above two rules. First, we apply the {\bf $T'$-deletion rule} if possible. Then we apply the {\bf $T'$-contraction rule} and if after the application of this rule we again can again apply the {\bf $T'$-deletion rule}, we do so.

For $i=0,1,2$, let $T_i=\{v\in T': |N_G(v)\cap W'|=i\}$, and $T_{\geq 3}=\{v\in T':~ |N_G(v)\cap W'|\geq 3\}$. Because we exhaustively applied the {\bf vertex deletion rule}, we have that
$T_0=T_1=\emptyset$. By Lemma~\ref{lem:bound-bip}, $|T_{\geq 3}|\leq 2|N_G(T')|-4\leq 2|W'|-4$ (or $T_{\geq 3}$ is empty). Therefore we have that~$G[T']$ has at most $2|W'|$ components that contain vertices of $T_{\geq 3}$. It remains to evaluate~$|T_2|$. Because of the {\bf vertex deletion rule}, for any $v\in T_2$, $d_G(v)-\delta(v)=1$ as otherwise we would either stop or delete the neighbours of~$v$ in~$W$.
Any two distinct $u,v\in T_2$ such that $N_G(u)\cap W'=N_G(v)\cap W'$ belong to distinct components of~$G[T']$ by the {\bf $T'$-deletion rule} and the {\bf $T'$-contraction rule}. There are at most~$\binom{|W'|}{2}$ such components that are isolated vertices of~$G[T']$ and there are at most $\binom{|W'|}{2}|T_{\geq 3}|$ vertices in~$T_2$ that there are in same components with the vertices of~$T_{\geq 3}$, and the total number of such vertices is at most $\binom{|W'|}{2}(2|W'|)$. Let~$T_2'$ denote the set of remaining vertices of~$T_2$.
Observe that each component of~$G[T_2']$ is a component of~$G[T']$ and has at least two vertices of~$T_2$. Moreover, for any two vertices~$u$ and~$v$ in the same component of~$G[T_2']$, $N_G(u)\cap W'\neq N_G(v)\cap W'$. Let~$G'$ be the graph obtained from~$G$ by contracting the edges of~$G[T_2']$. Each component of~$G[T_2']$ is contracted into a single vertex. Let~$Z$ denote the set of vertices of~$G'$ obtained from the components of~$G[T_2']$. The set~$Z$ is independent and for each $v\in Z$, $d_{G'}(v)\geq 3$. By Lemma~\ref{lem:bound-bip}, $|Z|\leq 2|N_{G'}(Z)|-4\leq 2|W'|-4$ (or~$Z$ is empty). Hence,~$G[T_2']$ has at most $2|W'|$ components. Because each component has at most~$\binom{|W'|}{2}$ vertices, $|T_2'|\leq \binom{|W|}{2}(2|W'|)$. Hence, $|T_2|\leq \binom{|W'|}{2}(4|W'|+1)$.
We have that $|V(G)|=|W'|+|T'|=|W'|+|T_0|+|T_1|+|T_2|+|T_{\geq 3}|=O(|W'|^3)$. Since~$W'$ has $(k_v+k_e)^{O(\alpha^2)}$ vertices, we obtain that the obtained graph~$G$ has size~$k^{O(1)}$ where $k=k_v+k_e$, i.e. we have a polynomial kernel.

To complete the proof, it remains to observe that the construction of the normalized instance can be done in polynomial time by Lemma~\ref{lem:normalization}, the construction of~$W$ and~$L$ can be done in polynomial time by Lemma~\ref{lem:tw}, and all the subsequent reduction rules can be applied in polynomial time.
\qed
\end{proof}

\section{Conclusions}

We proved that \textsc{DPGGD} and \textsc{DCPGGD} are \classNP-complete but allow polynomial kernels when parameterized by $k_v+k_e$. These problems generalize
the {\sc Degree Constrained Editing($S$)} problem and its connected variant for
$S=\{\ed,\vd\}$; this can be seen, for instance, by testing all possible pairs $k_v,k_e$ with $k_v+k_e=k$ or by a slight adjustment of our algorithms.
Note that by setting $k_v=0$ or $k_e=0$ we obtain the same results
for $S=\{\ed\}$ and $S=\{\vd\}$, respectively (recall though that for $S=\{\ed\}$ this is not so surprising, as the less general problem {\sc Degree Constrained Editing($\{\ed\}$)}
is polynomial-time solvable for general graphs).

Several open problems remain. We note that graph modification problems that permit edge additions are less natural to consider for planar graphs, because the class of planar graphs is not closed under edge
addition.
However, we could allow other, more appropriate, operations such as edge contractions and vertex dissolutions when considering planar graphs.
Belmonte et al.~\cite{BGHP14} considered the setting in which only edge contractions are allowed and
obtained initial results for general graphs that extend the work of Mathieson and Szeider~\cite{MathiesonS12}
on {\sc Degree Constrained Editing($S$)} in this direction.

\bibliographystyle{splncs03}
\bibliography{Graph_edit}

\begin{thebibliography}{10}
\providecommand{\url}[1]{\texttt{#1}}
\providecommand{\urlprefix}{URL }

\bibitem{BGHP14}
Belmonte, R., Golovach, P.A., van~'t Hof, P., Paulusma, D.: Parameterized
  complexity of three edge contraction problems with degree constraints. Acta
  Informatica  51(7),  473--497 (2014)

\bibitem{BodlaenderDDFLP13}
Bodlaender, H.L., Drange, P.G., Dregi, M.S., Fomin, F.V., Lokshtanov, D.,
  Pilipczuk, M.: An {O({c\^{}}k n)} 5-approximation algorithm for treewidth.
  In: {FOCS} 2013. pp. 499--508. {IEEE} Computer Society (2013)

\bibitem{BodlaenderFLPST09}
Bodlaender, H.L., Fomin, F.V., Lokshtanov, D., Penninkx, E., Saurabh, S.,
  Thilikos, D.M.: (meta) kernelization. In: {FOCS} 2009. pp. 629--638. {IEEE}
  Computer Society (2009)

\bibitem{BodlaenderFLPST09a}
Bodlaender, H.L., Fomin, F.V., Lokshtanov, D., Penninkx, E., Saurabh, S.,
  Thilikos, D.M.: (meta) kernelization. CoRR  abs/0904.0727 (2009)

\bibitem{BoeschST77}
Boesch, F.T., Suffel, C.L., Tindell, R.: The spanning subgraphs of {Eulerian}
  graphs. Journal of Graph Theory  1(1),  79--84 (1977)

\bibitem{BurzynBD06}
Burzyn, P., Bonomo, F., Dur{\'a}n, G.: {NP}-completeness results for edge
  modification problems. Discrete Applied Mathematics  154(13),  1824--1844
  (2006)

\bibitem{Cai96}
Cai, L.: Fixed-parameter tractability of graph modification problems for
  hereditary properties. Inf. Process. Lett.  58(4),  171--176 (1996)

\bibitem{CaiY11}
Cai, L., Yang, B.: Parameterized complexity of even/odd subgraph problems.
  Journal of Discrete Algorithms  9(3),  231--240 (2011)

\bibitem{CyganMPPS14}
Cygan, M., Marx, D., Pilipczuk, M., Pilipczuk, M., Schlotter, I.: Parameterized
  complexity of {Eulerian} deletion problems. Algorithmica  68(1),  41--61
  (2014)

\bibitem{DGHP14}
Dabrowski, K.K., Golovach, P.A., van~'t Hof, P., Paulusma, D.: Editing to
  {Eulerian} graphs. In: {FSTTCS} 2014. LIPIcs, vol.~29, pp. 97--108. Schloss
  Dagstuhl - Leibniz-Zentrum fuer Informatik (2014)

\bibitem{DGHPT15}
Dabrowski, K.K., Golovach, P.A., van~'t Hof, P., Paulusma, D., Thilikos, D.M.:
  Editing to a planar graph of given degrees. In: {CSR} 2015. Lecture Notes in
  Computer Science, vol. 9139, pp. 143--156. Springer (2015)

\bibitem{DowneyF13}
Downey, R.G., Fellows, M.R.: Fundamentals of Parameterized Complexity. Texts in
  Computer Science, Springer (2013)

\bibitem{FlumG06}
Flum, J., Grohe, M.: Parameterized complexity theory. Texts in Theoretical
  Computer Science. An EATCS Series, Springer-Verlag, Berlin (2006)

\bibitem{FominLST12}
Fomin, F.V., Lokshtanov, D., Saurabh, S., Thilikos, D.M.: Linear kernels for
  (connected) dominating set on \emph{H}-minor-free graphs. In: {SODA} 2012.
  pp. 82--93. {SIAM} (2012)

\bibitem{FroeseNN14}
Froese, V., Nichterlein, A., Niedermeier, R.: Win-win kernelization for degree
  sequence completion problems. In: {SWAT} 2014. Lecture Notes in Computer
  Science, vol. 8503, pp. 194--205. Springer (2014)

\bibitem{GareyJT76}
Garey, M.R., Johnson, D.S., Tarjan, R.E.: The planar hamiltonian circuit
  problem is {NP}-complete. {SIAM} J. Comput.  5(4),  704--714 (1976)

\bibitem{GarneroPST14}
Garnero, V., Paul, C., Sau, I., Thilikos, D.M.: Explicit linear kernels via
  dynamic programming. In: {STACS} 2014. LIPIcs, vol.~25, pp. 312--324. Schloss
  Dagstuhl - Leibniz-Zentrum fuer Informatik (2014)

\bibitem{GarneroST14}
Garnero, V., Sau, I., Thilikos, D.M.: A linear kernel for planar red-blue
  dominating set. CoRR  abs/1408.6388 (2014)

\bibitem{Golovach14}
Golovach, P.A.: Editing to a connected graph of given degrees. In: {MFCS} 2014,
  Part {II}. Lecture Notes in Computer Science, vol. 8635, pp. 324--335.
  Springer (2014)

\bibitem{Golovach14a}
Golovach, P.A.: Editing to a graph of given degrees. Theor. Comput. Sci.  591,
  72--84 (2015)

\bibitem{HT74}
Hopcroft, J.E., Tarjan, R.E.: Efficient planarity testing. J. {ACM}  21(4),
  549--568 (1974)

\bibitem{KhotR02}
Khot, S., Raman, V.: Parameterized complexity of finding subgraphs with
  hereditary properties. Theor. Comput. Sci.  289(2),  997--1008 (2002)

\bibitem{KimLPRRSS13}
Kim, E.J., Langer, A., Paul, C., Reidl, F., Rossmanith, P., Sau, I., Sikdar,
  S.: Linear kernels and single-exponential algorithms via protrusion
  decompositions. In: {ICALP} 2013. Lecture Notes in Computer Science, vol.
  7965, pp. 613--624. Springer (2013)

\bibitem{Kloks94}
Kloks, T.: Treewidth, Computations and Approximations, Lecture Notes in
  Computer Science, vol. 842. Springer (1994)

\bibitem{LewisY80}
Lewis, J.M., Yannakakis, M.: The node-deletion problem for hereditary
  properties is {NP}-complete. J. Comput. Syst. Sci.  20(2),  219--230 (1980)

\bibitem{MathiesonS12}
Mathieson, L., Szeider, S.: Editing graphs to satisfy degree constraints: A
  parameterized approach. J. Comput. Syst. Sci.  78(1),  179--191 (2012)

\bibitem{MoserT09}
Moser, H., Thilikos, D.M.: Parameterized complexity of finding regular induced
  subgraphs. J. Discrete Algorithms  7(2),  181--190 (2009)

\bibitem{NatanzonSS01}
Natanzon, A., Shamir, R., Sharan, R.: Complexity classification of some edge
  modification problems. Discrete Applied Mathematics  113(1),  109--128 (2001)

\bibitem{Niedermeierbook06}
Niedermeier, R.: Invitation to fixed-parameter algorithms, Oxford Lecture
  Series in Mathematics and its Applications, vol.~31. Oxford University Press,
  Oxford (2006)

\bibitem{Stewart94}
Stewart, I.A.: Deciding whether a planar graph has a cubic subgraph is
  {NP}-complete. Discrete Mathematics  126(1-3),  349--357 (1994)

\bibitem{Yannakakis78}
Yannakakis, M.: Node- and edge-deletion {NP}-complete problems. In: {STOC}
  1978. pp. 253--264. ACM (1978)

\end{thebibliography}
\end{document}